\documentclass[12pt]{article}

\usepackage[margin=1in]{geometry}
\usepackage{amsthm}
\usepackage{amssymb}
\usepackage{amsmath}
\usepackage{graphicx}
\usepackage{mathdots}
\usepackage{mathtools}
\usepackage{url}
\usepackage{color}
\usepackage{framed}
\usepackage[ruled,vlined]{algorithm2e}

\newtheorem{theorem}{Theorem}%[section]

\newtheorem{lemma}[theorem]{Lemma}
\newtheorem{proposition}[theorem]{Proposition}

\theoremstyle{definition}
\newtheorem{problem}[theorem]{Problem}

\title{Sketching with Kerdock's crayons:\\ Fast sparsifying transforms for arbitrary linear maps}

\author{
Tim~Fuchs\footnote{Department of Mathematics, Technical University of Munich, M\"{u}nchen, Germany}
\qquad
David~Gross\footnote{Institute for Theoretical Physics, University of Cologne, Cologne, Germany}
\qquad
Felix~Krahmer\footnotemark[1]
\vspace{0.1in}\\
Richard~Kueng\footnote{Institute for Integrated Circuits, Johannes Kepler University Linz, Linz, Austria}
\qquad
Dustin~G.~Mixon\footnote{Department of Mathematics, The Ohio State University, Columbus, Ohio, USA} \footnote{Translational Data Analytics Institute, The Ohio State University, Columbus, Ohio, USA}
}

\date{}

\begin{document}
\maketitle

\begin{abstract}
Given an arbitrary matrix $A\in\mathbb{R}^{n\times n}$, we consider the fundamental problem of computing $Ax$ for any $x\in\mathbb{R}^n$ such that $Ax$ is $s$-sparse.
While fast algorithms exist for particular choices of $A$, such as the discrete Fourier transform, there is currently no $o(n^2)$ algorithm that treats the unstructured case.
In this paper, we devise a randomized approach to tackle the unstructured case. 
Our method relies on a representation of $A$ in terms of certain real-valued mutually unbiased bases derived from Kerdock sets. 
In the preprocessing phase of our algorithm, we compute this representation of $A$ in $O(n^3\log n)$ operations.
Next, given any unit vector $x\in\mathbb{R}^n$ such that $Ax$ is $s$-sparse, our randomized fast transform uses this representation of $A$ to compute the entrywise $\epsilon$-hard threshold of $Ax$ with high probability in only $O(sn + \epsilon^{-2}\|A\|_{2\to\infty}^2n\log n)$ operations.
In addition to a performance guarantee, we provide numerical results that demonstrate the plausibility of real-world implementation of our algorithm.
\end{abstract}

\section{Introduction}

Computing matrix--vector products is a fundamental part of numerical linear algebra.
The naive algorithm takes $O(mn)$ operations to multiply an $m\times n$ matrix by a vector.
Many structured matrices admit a more efficient implementation of this computation, the most well-known example being the fast Fourier transform, which takes only $O(n\log n)$ operations.
In some applications, the desired Fourier transform of the given signal is nearly $s$-sparse, and as we discuss below, a number of works have proposed methods for such cases that are sublinear in the dimension $n$.

For the one-dimensional discrete Fourier transform, a randomized algorithm with a runtime scaling quadratically in $s$ up to logarithmic factors in the dimension $n$ has been provided in \cite{Gilbert:02}, while a deterministic approach with similar complexity was found in \cite{Iwen:10, Iwen:13}.
In later works, this could be reduced to linear scaling in $s$ for both random \cite{Gilbert:05, Iwen:10, Al-Hassanieh:12, Al-Hassanieh:12b, Iwen:13} and deterministic \cite{Lawlor:13, Christlieb:16} algorithms.
For the $d$-dimensional Fourier transform applied to signals in $n=N^d$ dimensions, the exponential scaling in $d$ presents an instance of the curse of dimensionality.
Despite this, for random signals, one may obtain runtimes that are linear in $sd$ up to logarithmic factors in $N$ \cite{Choi:21, Choi:19}.
For deterministic signals, various deterministic \cite{Morotti:17, Iwen:13} and random \cite{Indyk:14, Kammerer:21, Choi:18, Choi:19b, Kammerer:20} sampling strategies have been proposed with a computational complexity which scales polynomially in $d$, $s$ and $N$ up to logarithmic factors.

Naturally, research on fast transforms is not restricted to Fourier structure.
For example, \cite{Potts:15} proposes a multidimensional Chebyshev transform with reduced runtime.
In \cite{Choi:18}, a more general approach has been established that yields fast sparse transforms for arbitrary bounded orthonormal product basis with a runtime scaling polynomially in $s$ up to logarithmic factors.
These results have been generalized in \cite{Choi:19b} to signals with only an approximately $s$-sparse representation while maintaining a computational complexity that is sublinear in the dimension $n$.
While covering a significantly larger class of transforms than just the Fourier transform, all these approaches remain restricted to a specific structure or class of structures of the transformation matrix.

At the same time, data-driven sparsifying transforms, which have been demonstrated to outperform predefined structured representation systems in a variety of contexts~\cite{EladA:06,BrucksteinDE:09,TosicF:11,MairalBP:12,RavishankarB:12,NamDEG:13}, typically do not have structural properties that allow for the application of any of the above fast transform methods. This issue was addressed in \cite{RusuGH:16,RusuT:17} by imposing structure amenable to fast transforms on the learned representation system $A$ so as to facilitate the computation of $Ax$. At the same time, this imposed structure significantly limits the space of admissible transforms, and the question remains whether a fast transform can also be constructed for learned representation systems beyond these restrictions.

In this paper, we consider cases where the desired product $Ax$ is approximately sparse for a matrix $A$ that does not follow any preset structural constraints, e.g., because it
is learned from data.
In particular, we assume $A$ is arbitrary.
Note that to compute the mapping $(A,x)\mapsto Ax$ for an arbitrary matrix $A$ and vector $x$, one must first read the input $(A,x)\in\mathbb{R}^{m\times n}\times\mathbb{R}^n$.
Since this already requires $\Theta(mn)$ operations, naive matrix--vector multiplication is optimally efficient when computing an individual matrix--vector product.
To obtain a speedup, we instead apply the same transform $x\mapsto Ax$ to a stream of vectors $x\in\mathbb{R}^n$, which models the setting of many applications.
Our approach will require some upfront preprocessing given $A$ in exchange for a much faster per-vector computation.
While little work has been done in this vein, there has been quite a bit of work on related problems, which we discuss below.

The first result following this strategy~\cite{Williams:07} concerns a matrix--vector multiplication algorithm over finite semirings (for general matrices and vectors, hence not assuming any kind of sparsity), which performs $O(n^{2+\epsilon})$ operations of preprocessing on an $n\times n$ matrix before multiplying with an arbitrary vector in $O(n^2/(\epsilon \log n)^2)$ operations. 
The first and (to our knowledge) only algorithm that achieves a comparable result for real $m\times n$ matrices is the \textit{mailman algorithm} introduced in~\cite{LibertyZ:09}. 
Provided the matrix contains only a constant number of distinct values, the algorithm takes $O(mn)$ operations of preprocessing and then takes $O(mn/\log(m+n))$ operations to multiply the preprocessed matrix with an arbitrary vector.

As an alternative, one might batch the stream of vectors into matrices and then perform matrix multiplication.
(Granted, such a batched computation is unacceptable for many applications.)
Research on matrix multiplication was initiated by the seminal work of Strassen~\cite{Strassen:69}, which multiplies two arbitrary $n\times n$ matrices in only $O(n^{2.808})$ operations (i.e., much faster than the naive $O(n^3)$ algorithm).
Later algorithms~\cite{CoppersmithW:90,DavieS:13,Williams:12,LeGall:14} improved this computational complexity to its currently best known scaling of $O(n^{2.373})$.
After dividing by the batch size, this gives a per-vector cost of $O(n^{1.373})$ operations.
However, we note that such algorithms are infeasible in practice.

A more feasible approach to matrix multiplication was proposed by Drineas, Kanan, and Mahoney~\cite{DrineasKH:06}.
They compute a random approximation of the desired product by multiplying two smaller matrices: one consisting of $k$ randomly selected columns of the first matrix $A$, and the other consisting of corresponding rows from the second matrix $B$.
With high probability, the Frobenius norm of the estimate error is $O(\|A\|_F\|B\|_F/\sqrt{k})$. 
Unfortunately, if $B$ represents a batch of column vectors, then this guarantee offers little control of the error in each vector. 
On the other hand, if $B$ represents a single column vector $b$, then  $\|Ab\|_2$ is typically much smaller than $\|A\|_F\|b\|_2$, so the resulting relative error is quite large even for relatively large values of $k$.

\subsection{Our approach}

%Let $\Sigma_s^m$ denote the set of $v\in\mathbb{R}^m$ such that $|\operatorname{supp}(v)|\leq s$, i.e., $v$ is $s$-sparse.

Given $A\in\mathbb{R}^{m\times n}$, let $\Sigma(A,s,\delta)$ denote the set of all unit vectors $x\in\mathbb{R}^n$ for which
\[
\inf\Big\{\|Ax-v\|_\infty:v\in\mathbb{R}^m,|\operatorname{supp}(v)|\leq s\Big\}
\leq\delta,
\]
i.e., $Ax$ is $\delta$-close to being $s$-sparse.
Given $\epsilon>0$, let $h_\epsilon\colon\mathbb{R}\to\mathbb{R}$ denote the $\epsilon$-hard thresholding function defined by $h_\epsilon(t):=t\cdot\mathbf{1}\{|t|\geq\epsilon\}$.
By abuse of notation, we apply $h_\epsilon$ to the entries of a vector $v$ by writing $h_\epsilon(v)$.
We seek to solve the following:

\begin{problem}
Given an arbitrary $A\in\mathbb{R}^{m\times n}$, $s\in\mathbb{N}$, and $\epsilon>\delta\geq0$, preprocess $A$ so that one may quickly compute $h_\epsilon(Ax)$ for any $x\in\Sigma(A,s,\delta)$.
\end{problem}

%Let $\Sigma(A,s)$ denote the set of all unit vectors $x\in\mathbb{R}^n$ for which $|\operatorname{supp}(Ax)|\leq s$, and let $h_\epsilon\colon\mathbb{R}\to\mathbb{R}$ denote the $\epsilon$-hard thresholding function defined by $h_\epsilon(x):=x\cdot1(|x|\geq\epsilon)$.
%By abuse of notation, we apply $h_\epsilon$ to the entries of a vector $v$ by writing $h_\epsilon(v)$.
%Given an arbitrary $A\in\mathbb{R}^{m\times n}$ and $\epsilon>0$, we seek to quickly compute $h_\epsilon(Ax)$ for any $x\in\Sigma(A,s)$.

%Let $\Sigma(n,s,\epsilon)$ denote the set of all $x\in\mathbb{R}^n$ such that 
%\[
%|\operatorname{supp}(Ax)|\leq s
%\qquad
%\text{and}
%\qquad
%|(Ax)_i|\geq \epsilon\|x\|_2
%\qquad
%\forall i\in\operatorname{supp}(Ax).
%\]
%Given an arbitrary matrix $A\in\mathbb{R}^{n\times n}$, we seek to quickly compute $Ax$ for any $x\in\Sigma(n,s,\epsilon)$.

Our approach uses a specially designed random vector $z\in\mathbb{R}^n$ such that $\mathbb{E}zz^\top=I$:
\begin{equation}
\label{eq.matrix mult by average}
Ax
=A(\mathbb{E}zz^\top)x
=\mathbb{E}Azz^\top x.
\end{equation}
Denote the random vector $y:=Azz^\top x\in\mathbb{R}^m$.
The fact that $\mathbb{E}y=Ax$ suggests a Monte Carlo approach to estimate $Ax$.
That is, we will approximate the true average $Ax$ with an estimator determined by $N$ independent samples.
To obtain a fast algorithm in this vein, we will select a distribution for $z$ and an estimator $\hat\mu$ for $\mathbb{E}y$ that together satisfy three properties:
\begin{itemize}
\item[(i)]
the distribution of $z$ is discrete with small support,
\item[(ii)]
$\hat\mu$ can be computed in linear time from independent realizations $\{y_j\}_{j\in[N]}$ of $y$, and
\item[(iii)]
for each $x\in\Sigma(A,s,\delta)$, 
$\|\hat\mu-Ax\|_\infty
<\frac{\epsilon-\delta}{2}$ 
with high probability, even for small $N$.
\end{itemize}
Indeed, if (i)--(iii) hold, then one may compute $h_\epsilon(Ax)$ using the following (fast) algorithm:

Let $\{s_\ell:\ell\in[L]\}\subseteq\mathbb{R}^n$ denote the support of the distribution of $z$.
Given $A\in\mathbb{R}^{m\times n}$, we run the preprocessing step of computing $\{As_\ell\}_{\ell\in[L]}$ in $O(Lmn)$ operations.
Granted, this is more expensive than computing $Ax$, but we only need to compute $\{As_\ell\}_{\ell\in[L]}$ once, while we expect to compute $Ax$ for a stream of $x$'s.
Next, given $x\in\Sigma(A,s,\delta)$, we draw independent realizations $\{z_j\}_{j\in[N]}$ of $z$.
Since we already computed $\{As_\ell\}_{\ell\in[L]}$, we may then compute the corresponding realizations $\{y_j\}_{j\in[N]}$ of $y$ in $O(N (m+n))$ operations.
Next, by (ii), we may compute $\hat\mu$ from $\{y_j\}_{j\in[N]}$ in $O(mN)$ operations.
Finally, let $S\subseteq[m]$ denote the indices of the $s$ entries of $\hat\mu$ of largest magnitude.
Then by (iii), it holds with high probability that $|\hat\mu_i|>\frac{\epsilon+\delta}{2}$ for every $i\in\operatorname{supp}(h_\epsilon(Ax))$ while $|\hat\mu_i|<\frac{\epsilon+\delta}{2}$ for every $i\in[m]$ such that $|(Ax)_i|\leq\delta$.
Since $x\in\Sigma(A,s,\delta)$ by assumption, it follows that $S\supseteq\operatorname{supp}(h_\epsilon(Ax))$, and so $A_Sx$ determines $h_\epsilon(Ax)$, which we compute in $O(sn)$ additional operations.
(Of course, $A_Sx$ might determine other entries in the support of $Ax$, and we would not discard this information in practice.)

To obtain (i)--(iii), we take $z$ to be uniformly distributed over an appropriately scaled $n$-dimensional projection of a projective $2$-design, and for $\hat\mu$, we partition $[N]$ into batches and compute the entrywise median of means of $\{y_j\}_{j\in[N]}$ over these batches.
The projective $2$-design allows us to control the variance of each entry of the random vector $y$; see Lemma~\ref{lem.variance bound}.
Next, the median-of-means estimator improves over the sample mean by being less sensitive to outliers in the small random sample $\{y_j\}_{j\in[N]}$.
Thanks to this behavior, we can get away with drawing only $N=O((\epsilon-\delta)^{-2}\|A\|_{2\to\infty}^2\log(m/\eta))$ samples, where the induced norm $\|A\|_{2\to\infty}$ equals the largest $\ell_2$-norm of the rows of $A$, and $\eta$ denotes the failure probability of the randomized algorithm; see Theorem~\ref{thm.median of means}.
As a bonus, the Kerdock set--based projective $2$-design we use enjoys a fast matrix--vector multiplication algorithm, yielding a preprocessing step of only $O(mn^2\log n+n^2\log^2n)$ operations despite having $L=\Theta(n^2)$; see Lemma~\ref{lem.kerdock preprocessing}.
See Algorithm~\ref{alg.fst} for a summary of our approach.

\begin{algorithm}[t]
\SetAlgoLined
\KwData{Parameters $\epsilon>\delta\geq0$, $s,J,K\in\mathbb{N}$, matrix $A\in\mathbb{R}^{m\times n}$, stream of $x\in\Sigma(A,s,\delta)$}
\KwResult{Entrywise $\epsilon$-hard threshold of matrix--vector products $h_\epsilon(Ax)$}

\medskip

\textit{Preprocessing step}\\
Let $d$ denote the smallest power of $2$ that is at least $n$, and put $L:=d(d/2+1)$\\
Let $\{u_\ell\}_{\ell\in[L]}$ denote a projective $2$-design for $\mathbb{R}^d$ arising from a Kerdock set\\
Put $s_\ell:=\sqrt{d}\Pi u_\ell$, where $\Pi\in\mathbb{R}^{n\times d}$ denotes projection onto the first $n$ coordinates\\
Use Lemma~\ref{lem.kerdock preprocessing} to compute $\{As_\ell\}_{\ell\in[L]}$

\medskip

\textit{Streaming step}\\
Draw $N:=JK$ indices $\{\ell_j\}_{j\in[N]}$ uniformly from $[L]$\\
Compute $y_j:=(As_{\ell_j})(s_{\ell_j}^\top x)$ for each $j\in[N]$\\
Compute the entrywise median of means $\hat\mu$ of $\{y_j\}_{j\in[N]}$ over $K$ batches of size $J$\\
Let $S\subseteq[m]$ denote the indices of the $s$ entries of $\hat\mu$ with largest magnitude\\
Compute $A_Sx$ and output the indices and values of entries with magnitude at least $\epsilon$

\caption{Fast sparsifying transform for an arbitrary linear map
\label{alg.fst}}
\end{algorithm}

To quickly evaluate the utility of this algorithm, consider the following model: 
$A$ is an arbitrary $n\times n$ orthogonal matrix, and $x$ is a random vector such that the entries of $Ax$ are drawn independently from the following mixture:
\[
e_i^\top Ax
\sim\left\{\begin{array}{cl}
\mathcal{N}(0,1)&\text{with probability }p\\
0&\text{with probability }1-p.
\end{array}\right.
\]
Then the expected size of the support of $Ax$ is $pn$.
In this model, our algorithm provides a speedup over naive matrix--vector multiplication in the regime
\[
1\prec pn\prec \frac{n}{\log n}.
\]
To see this, first put $s:=10pn$ (say).
Then the multiplicative Chernoff bound implies that $Ax$ is $s$-sparse with high probability, and so we take $\delta=0$.
Before selecting $\epsilon>0$, we normalize our vector so that $\hat{x}:=x/\|x\|_2\in\Sigma(A,s,\delta)$ with high probability.
Next, standard tail bounds imply that $\|x\|_2^2\in\Theta(pn)$ with high probability.
This suggests the scaling $\epsilon:=\alpha/\sqrt{pn}$ with $\alpha\in(0,1)$.
Considering a fraction $\Theta(\alpha)$ of the support of $Ax$ has magnitude $\Theta(\alpha)$, the hard threshold $h_\epsilon(A\hat{x})$ serves as a decent estimate for the product $A\hat{x}$:
\[
\|h_\epsilon(A\hat{x})-A\hat{x}\|_2^2
=\Theta(\alpha^3).
\]
Furthermore, we obtain this quality of estimate with relatively little computation:
Since $m=n$, we have $O(n^3\log n)$ operations of preprocessing, and then for each $x$, we compute $h_\epsilon(A\hat{x})$ in $O(\alpha^{-2}sn\log n)$ operations since $\|A\|_{2\to\infty}=1$.
By comparison, if an oracle were to reveal the support of $h_\epsilon(A\hat{x})$, then naive matrix--vector multiplication with the appropriate submatrix of $A$ would cost $O(sn)$ operations.

\subsection{Outline}

In the next section, we review the necessary theory of projective designs, and we show how they can be used in conjunction with a median-of-means estimator to obtain a high-quality random estimate of a sparse matrix--vector product.
Next, Section~\ref{sec.kerdock} provides the details of a specific choice of projective design, namely, one that arises from a Kerdock set described by Calderbank, Cameron, Kantor, and Seidel~\cite{CalderbankCKS:97}.
This particular choice of projective design allows us to leverage the fast Walsh--Hadamard transform to substantially speed up the preprocessing step of our algorithm.
We conclude in Section~\ref{sec.numerics} with some numerical results that demonstrate the plausibility of a real-world implementation of our algorithm.

\section{Projective designs and the median of means}
\label{sec.designs median}

Let $\sigma$ denote the uniform probability measure on the unit sphere $S^{d-1}$ in $\mathbb{R}^d$, let $\operatorname{Hom}_{j}(\mathbb{R}^d)$ denote the set of homogeneous polynomials of total degree $j$ in $d$ real variables, and put
\[
c_{d,k}
:=\frac{1\cdot3\cdot5\cdots(2k-1)}{d(d+2)\cdots(d+2(k-1))}.
\]
A \textbf{projective $t$-design} for $\mathbb{R}^d$ is defined to be any $\{u_\ell\}_{\ell\in[L]}$ in $S^{d-1}$ that satisfies the following equivalent properties:

%\begin{definition}
%We say unit vectors $\{u_\ell\}_{\ell\in[L]}$ in $\mathbb{R}^d$ form a \textbf{projective $t$-design} if for every $k\in\{0,\ldots,t\}$ and every homogeneous polynomial function $p\colon\mathbb{R}^d\to\mathbb{R}$ of degree $2k$,
%\[
%\frac{1}{L}\sum_{\ell\in[L]}p(u_\ell)
%=\int_{S^{d-1}}p(u)d\sigma(u),
%\]
%where $\sigma$ denotes the uniform probability measure on the unit sphere $S^{d-1}$.
%\end{definition}

\begin{proposition}
\label{prop.proj t design}
Given $\{u_\ell\}_{\ell\in[L]}$ in $S^{d-1}$ and $t\in\mathbb{N}$, the following are equivalent:
\begin{itemize}
\item[(a)]
$\frac{1}{L}\sum_{\ell\in[L]}p(u_\ell)=\int_{S^{d-1}}p(u)d\sigma(u)$ for every $p\in\operatorname{Hom}_{2k}(\mathbb{R}^d)$ and every $k\in\{0,\ldots,t\}$.
\item[(b)]
$\frac{1}{L}\sum_{\ell\in[L]}\langle x,u_\ell\rangle^{2k}=c_{d,k}\|x\|^{2k}$ for every $x\in\mathbb{R}^d$ and every $k\in\{1,\ldots,t\}$.
\item[(c)]
$\frac{1}{L^2}\sum_{\ell,\ell'\in[L]}\langle u_\ell,u_{\ell'}\rangle^{2k}=c_{d,k}$ for every $k\in\{1,\ldots,t\}$.
\end{itemize}
\end{proposition}

The proof of Proposition~\ref{prop.proj t design} is contained in Section~6.4 of~\cite{Waldron:18} and references therein, given the observation that $\{u_\ell\}_{\ell\in[L]}$ satisfies Proposition~\ref{prop.proj t design}(a) precisely when $\{u_\ell\}_{\ell\in[L]}\cup\{-u_\ell\}_{\ell\in[L]}$ forms a so-called \textit{spherical $2t$-design}.
We will see that the cubature rule in Proposition~\ref{prop.proj t design}(a) is what makes projective $t$-designs useful, while Proposition~\ref{prop.proj t design}(c) makes them easy to identify.
We note that an analog of Proposition~\ref{prop.proj t design}(c) is used to define projective $t$-designs in a variety of settings, such as complex projective space and the Cayley plane~\cite{Munemasa:07}.

Our application of projective $2$-designs encourages us to take the size $L$ to be as small as possible.
To this end, there is a general lower bound~\cite{BannaiH:85} of $L\geq\binom{d+1}{2}$, but to date, equality is only known to be achieved for $d\in\{2,3,7,23\}$; see~\cite{Gillespie:18} and references therein.
Despite this scarcity, there are \textit{infinite} families of projective $2$-designs that take $L$ to be slightly larger, specifically, $L=d(d/2+1)$ whenever $d$ is a power of $4$; see~\cite{CameronS:91,CalderbankCKS:97}.
For these constructions, $\{u_\ell\}_{\ell\in[L]}$ takes the form of a union of orthonormal bases.
Orthonormal bases $\{x_i\}_{i\in[d]}$ and $\{y_i\}_{i\in[d]}$ are said to be \textbf{unbiased} if $|\langle x_i,y_j\rangle|^2=1/d$ for every $i,j\in[d]$.

\begin{proposition}
\label{prop.mubs}
Suppose $\{u_{b,i}\}_{i\in[d]}$ in $\mathbb{R}^d$ is orthonormal for every $b\in[d/2+1]$, and suppose further that $\{u_{b,i}\}_{i\in[d]}$ and $\{u_{b',i}\}_{i\in[d]}$ are unbiased for every $b,b'\in[d/2+1]$ with $b\neq b'$.
Then $\{u_{b,i}\}_{b\in[d/2+1],i\in[d]}$ forms a projective $2$-design for $\mathbb{R}^d$.
\end{proposition}

The proof of Propostion~\ref{prop.mubs} follows from Propostion~\ref{prop.proj t design}(c) and the definition of unbiased.
In Section~\ref{sec.kerdock}, we will provide an explicit construction of this form.
In the meantime, we show how projective $2$-designs are useful in our application.
The following result defines the random vector $z$ in terms of a projective $2$-design, and then uses this structure to control the variance of each coordinate of the random vector $y:=Azz^\top x$.

\begin{lemma}
\label{lem.variance bound}
Given $A\in\mathbb{R}^{m\times n}$, fix a projective $2$-design $\{u_\ell\}_{\ell\in[L]}$ for $\mathbb{R}^{d}$ with $d\geq n$, let $\Pi\colon\mathbb{R}^{d}\to\mathbb{R}^n$ denote projection onto the first $n$ coordinates, and let $z$ denote a random vector with uniform distribution over $\{\sqrt{d}\Pi u_\ell:\ell\in[L]\}$.
Given a unit vector $x\in\mathbb{R}^n$, define the random vector $y:=Azz^\top x\in\mathbb{R}^m$.
Then for each $i\in[m]$, it holds that
\[
\mathbb{E}[e_i^\top y]
=e_i^\top Ax,
\qquad
\operatorname{Var}(e_i^\top y)
\leq 2\|A\|_{2\to\infty}^2.
\]
\end{lemma}

\begin{proof}
Let $\Pi^*$ denote the adjoint of $\Pi$, namely, the map that embeds $\mathbb{R}^n$ into the first $n$ coordinates of $\mathbb{R}^{d}$.
Fix $i\in[m]$, let $a_i^\top$ denote the $i$th row of $A$, and put $\tilde{a}_i:=\Pi^*a_i$ and $\tilde{x}:=\Pi^*x$.
Then
\[
e_i^\top y
=e_i^\top Azz^\top x
=a_i^\top (\sqrt{d}\Pi u_{\ell(z)})(\sqrt{d}\Pi u_{\ell(z)})^\top x
=d\tilde{a}_i^\top u_{\ell(z)}u_{\ell(z)}^\top \tilde{x}.
\]
To compute $\mathbb{E}[e_i^\top y]$ and $\mathbb{E}[(e_i^\top y)^2]$, we will consider a random vector $u$ that is uniformly distributed on the unit sphere $S^{d-1}$, as well as some rotation $Q\in\operatorname{O}(d)$ such that
\[
Q\tilde{x}=e_1,
\qquad
Q\tilde{a}_i=\alpha_1e_1+\alpha_2e_2,
\qquad
\alpha_1,\alpha_2\in\mathbb{R}.
\]
Since $v\mapsto d\tilde{a}_i^\top vv^\top \tilde{x}$ resides in $\operatorname{Hom}_2(\mathbb{R}^d)$, we may apply Proposition~\ref{prop.proj t design}(a) to get
\[
\mathbb{E}[e_i^\top y]
=\frac{1}{L}\sum_{\ell\in[L]}d\tilde{a}_i^\top u_{\ell}u_{\ell}^\top \tilde{x}
=\mathbb{E}[d\tilde{a}_i^\top uu^\top \tilde{x}]
=\mathbb{E}[d\tilde{a}_i^\top Q^\top Quu^\top Q^\top Q \tilde{x}].
\]
A change of variables $Qu\mapsto u$ then gives
\[
\mathbb{E}[e_i^\top y]
=\mathbb{E}[d(Q\tilde{a}_i)^\top uu^\top (Q \tilde{x})]
=\mathbb{E}[d(\alpha_1e_1+\alpha_2e_2)^\top uu^\top e_1]
=d\alpha_1\mathbb{E}u_1^2.
\]
Considering $1=\mathbb{E}[\|u\|^2]=d\mathbb{E}u_1^2$ by symmetry and linearity of expectation, it follows that
\[
\mathbb{E}[e_i^\top y]
=\alpha_1
=(\alpha_1e_1+\alpha_2e_2)^\top e_1
=\tilde{a}_i^\top\tilde{x}
=a_i^\top x
=e_i^\top Ax.
\]
Indeed, this behavior was the original motivation~\eqref{eq.matrix mult by average} for our approach.
We will apply the same technique to compute $\mathbb{E}[(e_i^\top y)^2]$.
Since $v\mapsto (d\tilde{a}_i^\top vv^\top \tilde{x})^2$ resides in $\operatorname{Hom}_4(\mathbb{R}^d)$, we may apply Proposition~\ref{prop.proj t design}(a) to get
\[
\mathbb{E}[(e_i^\top y)^2]
=\frac{1}{L}\sum_{\ell\in[L]}(d\tilde{a}_i^\top u_\ell u_\ell^\top \tilde{x})^2
=\mathbb{E}[(d\tilde{a}_i^\top uu^\top \tilde{x})^2]
=\mathbb{E}[(d\tilde{a}_i^\top Q^\top Quu^\top Q^\top Q\tilde{x})^2].
\]
A change of variables $Qu\mapsto u$ then gives
\[
\mathbb{E}[(e_i^\top y)^2]
=\mathbb{E}[(d(\alpha_1 e_1+\alpha_2 e_2)^\top uu^\top e_1)^2]
=d^2(\alpha_1^2\mathbb{E}u_1^4+\alpha_2^2\mathbb{E}u_1^2u_2^2).
\]
Next, the theorem in~\cite{Folland:01} implies $\mathbb{E}u_1^4=\frac{3}{d(d+2)}$ and $\mathbb{E}u_1^2u_2^2=\frac{1}{d(d+2)}$, thereby implying
\[
\mathbb{E}[(e_i^\top y)^2]
=\alpha_1^2\cdot\frac{3d}{d+2}+\alpha_2^2\cdot\frac{d}{d+2}.
\]
Finally, we recall $\mathbb{E}[e_i^\top y]=\alpha_1$ to compute the desired variance:
\[
\operatorname{Var}(e_i^\top y)
=\mathbb{E}[(e_i^\top y)^2]-(\mathbb{E}[e_i^\top y])^2
=2\alpha_1^2\cdot\frac{d-1}{d+2}+\alpha_2^2\cdot\frac{d}{d+2}
\leq 2\alpha_1^2+\alpha_2^2
\leq 2\|a_i\|^2.
\]
The result then follows from the fact that $\|A\|_{2\to\infty}^2=\max_{i\in[m]}\|a_i\|^2$.
\end{proof}

Now that we have control of the variance, we can obtain strong deviation bounds on our median-of-means estimator $\hat\mu$ of $\mathbb{E}y=Ax$.

\begin{theorem}
\label{thm.median of means}
Given $A\in\mathbb{R}^{m\times n}$, fix a projective $2$-design $\{u_\ell\}_{\ell\in[L]}$ for $\mathbb{R}^{d}$ with $d\geq n$, let $\Pi\colon\mathbb{R}^{d}\to\mathbb{R}^n$ denote projection onto the first $n$ coordinates, and let $z$ denote a random vector with uniform distribution over $\{\sqrt{d}\Pi u_\ell:\ell\in[L]\}$.
Given a unit vector $x\in\mathbb{R}^n$, define the random vector $y:=Azz^\top x\in\mathbb{R}^m$, select
\[
J\geq\frac{4e^2\|A\|_{2\to\infty}^2}{\gamma^2},
\qquad
K\geq2\log\Big(\frac{m}{\eta}\Big),
\]
draw independent copies $\{y_{jk}\}_{j\in[J],k\in[K]}$ of $y$, and compute the entrywise median of means:
\[
\hat\mu:=\operatorname{median}\{\overline{y}_k\}_{k\in[K]},
\qquad
\overline{y}_k:=\frac{1}{J}\sum_{j\in[J]}y_{jk}.
\]
Then $\|\hat\mu-Ax\|_\infty<\gamma$ with probability at least $1-\eta$.
\end{theorem}

\begin{proof}
Fix $i\in[m]$.
For notational convenience, we denote the random variable $Y:=e_i^\top y$.
Given independent copies $\{Y_j\}_{j\in[J]}$ of $Y$, let $\overline{Y}$ denote their sample average.
Then Chebyshev's inequality and Lemma~\ref{lem.variance bound} together imply the deviation inequality
\begin{equation}
\label{eq.chebyshev}
p
:=\mathbb{P}\{|\overline{Y}-(Ax)_i|\geq\gamma\}
=\mathbb{P}\{|\overline{Y}-\mathbb{E}\overline{Y}|\geq\gamma\}
\leq\frac{\operatorname{Var}(\overline{Y})}{\gamma^2}
\leq \frac{2\|A\|_{2\to\infty}^2\|x\|^2}{J\gamma^2}.
\end{equation}
Now take independent copies $\{\overline{Y}_k\}_{k\in[K]}$ of $\overline{Y}$ and put $\hat\mu_i:=\operatorname{median}\{\overline{Y}_k\}_{k\in[K]}$.
Notice that $\hat\mu_i\geq(Ax)_i+\gamma$ only if half of the $\overline{Y}_k$'s satisfy $\overline{Y}_k\geq (Ax)_i+\gamma$.
Similarly, $\hat\mu_i\leq(Ax)_i-\gamma$ only if half of the $\overline{Y}_k$'s satisfy $\overline{Y}_k\leq (Ax)_i-\gamma$.
Thus, defining $I_k:=\mathbf{1}\{|\overline{Y}_k-(Ax)_i|\geq\gamma\}$, we have
\[
\mathbb{P}\{|\hat\mu_i-(Ax)_i|\geq\gamma\}
\leq\mathbb{P}\bigg\{\sum_{k\in[K]}I_k\geq\frac{K}{2}\bigg\}.
\]
Since $\{I_k\}_{k\in[K]}$ are independent Bernoulli random variables with success probability $p$, we may continue with the help of the multiplicative Chernoff bound:
\[
\mathbb{P}\bigg\{\sum_{k\in[K]}I_k\geq(1+\lambda)Kp\bigg\}
\leq\Big(\frac{e^\lambda}{(1+\lambda)^{1+\lambda}}\Big)^{Kp}
=e^{-Kp}\Big(\frac{e}{1+\lambda}\Big)^{(1+\lambda)Kp},
\qquad
\lambda>0.
\]
By our choice of $J$, equation \eqref{eq.chebyshev} implies that $p\leq 1/(2e^2)<1/2$.
As such, there exists $\lambda>0$ such that $(1+\lambda)Kp=K/2$.
Combining the above bounds then gives
\[
\mathbb{P}\{|\hat\mu_i-(Ax)_i|\geq\gamma\}
\leq e^{-Kp}(2ep)^{K/2}
\leq(2ep)^{K/2}
\leq\Big(\frac{4e\|A\|_{2\to\infty}^2\|x\|^2}{J\gamma^2}\Big)^{K/2}
\leq e^{-K/2}
\leq\frac{\eta}{m},
\]
where the last steps follow from our choices for $J$ and $K$.
Finally, since our choice for $i\in[m]$ was arbitrary, the result follows from a union bound.
\end{proof}

\section{Fast preprocessing with Kerdock sets}
\label{sec.kerdock}

Select $k\in2\mathbb{N}$, and consider the real vector space $\mathbb{R}[\mathbb{F}_2^k]$ of functions $f\colon\mathbb{F}_2^k\to\mathbb{R}$.
Calderbank, Cameron, Kantor, and Seidel~\cite{CalderbankCKS:97} describe a projective $2$-design in this space that takes the form of mutually unbiased orthonormal bases (\`{a} la Proposition~\ref{prop.mubs}).
We explicitly construct this projective $2$-design with some help from the underlying finite field, and then we leverage its structure to speed up the preprocessing step of our algorithm.

We say $M\in\mathbb{F}_2^{k\times k}$ is \textbf{skew-symmetric} if $M_{ii}=0$ and $M_{ij}=M_{ji}$ for every $i,j\in[k]$.
Given a skew-symmetric $M$, consider the corresponding upper-triangular matrix $\tilde{M}\in\mathbb{F}_2^{k\times k}$ defined by $\tilde{M}_{ij}:=M_{ij}\cdot\mathbf{1}\{i<j\}$ and the quadratic form $Q_M\colon\mathbb{F}_2^k\times\mathbb{F}_2^k\to\mathbb{F}_2$ defined by $Q_M(x):=x^\top\tilde{M} x$.
Given a skew-symmetric $M\in\mathbb{F}_2^{k\times k}$ and $w\in\mathbb{F}_2^k$, define $u_{M,w}\in\mathbb{R}[\mathbb{F}_2^k]$ by
\begin{equation}
\label{eq.definition of vectors}
u_{M,w}(x)
:=2^{-k/2}(-1)^{Q_M(x)+w^\top x}.
\end{equation}
%Orthonormal bases $\{x_i\}_{i\in[d]}$ and $\{y_i\}_{i\in[d]}$ for a common real vector space are said to be \textbf{unbiased} if $|\langle x_i,y_j\rangle|^2=1/d$ for every $i,j\in[d]$.
The following result uses the vectors in \eqref{eq.definition of vectors} to form unbiased orthonormal bases for $\mathbb{R}[\mathbb{F}_2^k]$; this result is contained in~\cite{Kantor:82a,Kantor:82b,CalderbankCKS:97}, but the proof is distributed over dozens of dense pages from multiple papers, so we provide a direct and illustrative proof at the end of this section.

\begin{proposition}
\label{prop.kerdock2mub}
Consider any skew-symmetric $M,M'\in\mathbb{F}_2^{k\times k}$.
\begin{itemize}
\item[(a)]
$\{u_{M,w}\}_{w\in\mathbb{F}_2^k}$ and $\{u_{M',w}\}_{w\in\mathbb{F}_2^k}$ are orthonormal bases.
\item[(b)]
If $M+M'$ has full rank over $\mathbb{F}_2$, then $\{u_{M,w}\}_{w\in\mathbb{F}_2^k}$ and $\{u_{M',w}\}_{w\in\mathbb{F}_2^k}$ are unbiased.
\end{itemize}
\end{proposition}

A \textbf{Kerdock set} is a collection $\mathcal{K}\subseteq\mathbb{F}_2^{k\times k}$ of $2^{k-1}$ skew-symmetric matrices such that $M+M'$ has full rank for every $M,M'\in\mathcal{K}$ with $M\neq M'$.
(Note that one cannot hope for a larger set with this property since the first row of each matrix in $\mathcal{K}$ must be distinct, and the first entry of these rows must equal zero.)
By Proposition~\ref{prop.kerdock2mub} and equation~\eqref{eq.definition of vectors}, a Kerdock set $\mathcal{K}$ determines $2^{k-1}+1$ mutually unbiased bases in $\mathbb{R}[\mathbb{F}_2^k]$, namely, $\{u_{M,w}\}_{w\in\mathbb{F}_2^k}$ for each $M\in\mathcal{K}$ and the identity basis $\{e_w\}_{w\in\mathbb{F}_2^k}$.
By Proposition~\ref{prop.mubs}, these bases combine to form a projective $2$-design.
Below, we give the ``standard'' Kerdock set described in Example~9.2 in~\cite{CalderbankCKS:97}.
(Note that \cite{CalderbankCKS:97} contains a typo and omits the proof of this result; specifically, when they write $\alpha x$, it should be $ax$; we will use $s$ instead of $a$ so as to clearly distinguish from $\alpha$, and for completeness, we supply a proof at the end of this section.)
In what follows, $\operatorname{tr}\colon\mathbb{F}_{2^{k-1}}\to\mathbb{F}_2$ denotes the field trace, while for any finite set $B$, we let $\mathbb{F}_2^{B\times B}$ denote the set of $|B|\times|B|$ matrices with entries in $\mathbb{F}_2$ whose rows and columns are indexed by $B$.

\begin{proposition}
\label{prop.kerdock construction}
Consider the $k$-dimensional vector space $V:=\mathbb{F}_{2^{k-1}}\times\mathbb{F}_2$ over the scalar field $\mathbb{F}_2$, and for each $s\in\mathbb{F}_{2^{k-1}}$, define the linear map $L_s\colon V\to V$ by
\[
L_s(x,\alpha)
:=(s^2x+s\operatorname{tr}(sx)+\alpha s,\operatorname{tr}(sx)).
\]
Next, select a basis $B$ for $V$, and consider the bilinear form
\[
(x,\alpha)\cdot(y,\beta)
:=\operatorname{tr}(xy)+\alpha\beta.
\]
Then $\{M_s\in\mathbb{F}_2^{B\times B}:s\in\mathbb{F}_{2^{k-1}}\}$ defined by $(M_s)_{b,b'}:=b\cdot L_s(b')$ is a Kerdock set.
\end{proposition}

Importantly, Kerdock sets provide speedups for the preprocessing step of our algorithm:

\begin{lemma}
\label{lem.kerdock preprocessing}
Given $A\in\mathbb{R}^{m\times n}$, select the smallest $k\in2\mathbb{N}$ satisfying $d:=2^k\geq n$.
Let $\mathcal{K}\subseteq\mathbb{F}_2^{k\times k}$ denote the Kerdock set described in Proposition~\ref{prop.kerdock construction}, consider $u_{M,w}\in\mathbb{R}[\mathbb{F}_2^k]$ defined by \eqref{eq.definition of vectors}, let $\{e_w\}_{w\in\mathbb{F}_2^k}$ denote the identity basis in $\mathbb{R}[\mathbb{F}_2^k]$, and let $\Pi\colon\mathbb{R}[\mathbb{F}_2^k]\to\mathbb{R}^n$ denote any coordinate projection.
Then
\[
\{A\sqrt{d}\Pi u_{M,w}\}_{M\in\mathcal{K},w\in\mathbb{F}_2^k}\cup\{A\sqrt{d}\Pi e_w\}_{w\in\mathbb{F}_2^k}
\]
can be computed in $O(mn^2\log n+n\log^2n)$ operations.
\end{lemma}

\begin{proof}
We identify each member of $\mathbb{R}[\mathbb{F}_2^k]$ as a vector in $\mathbb{R}^d$ with entries indexed by $\mathbb{F}_2^k$.
By this identification, the vectors $\{u_{M,w}\}_{w\in\mathbb{F}_2^k}$ appear as the columns of the matrix product $D_MH^{\otimes k}$, where $\cdot^{\otimes k}$ denotes the Kronecker power and $D_M,H\in\mathbb{R}^{d\times d}$ are defined by
\[
D_M:=\operatorname{diag}\{(-1)^{Q_M(x)}\}_{x\in\mathbb{F}_2^k},
\qquad
H_{ab}:=2^{-1/2}(-1)^{ab},
\qquad
a,b\in\mathbb{F}_2.
\]
Let $a_i^\top$ denote the $i$th row of $A$.
Then the following algorithm runs in the claimed number of operations:
For each $M\in\mathcal{K}$, (i) compute $\{(-1)^{Q_M(x)}\}_{x\in\mathbb{F}_2^k}$ in $O(n\log^2n)$ operations, (ii) for each $i\in[m]$, use the fast Walsh--Hadamard transform to compute $((\Pi^*a_i)^\top D_M)H^{\otimes k}$ in $O(n\log n)$ operations, and (iii) compute the columns of
\[
\sqrt{d}\sum_{i\in[m]}e_i(a_i^\top \Pi D_MH^{\otimes k})
=A\sqrt{d}\Pi D_MH^{\otimes k}
\]
in $O(mn)$ operations; finally, compute $\{A\sqrt{d}\Pi e_w\}_{w\in\mathbb{F}_2^k}$ in $O(mn)$ operations.
\end{proof}

\begin{proof}[Proof of Proposition~\ref{prop.kerdock2mub}]
Define the \textit{Heisenberg group} $H_{2k+1}(\mathbb{F}_2)$ to be the set $\mathbb{F}_2^k\times\mathbb{F}_2^k\times\mathbb{F}_2$ with multiplication
\[
(p,q,\epsilon)\cdot(p',q',\epsilon')
:=(p+p',q+q',\epsilon+\epsilon'+q^\top p').
\]
%Observe that $(p,q,\epsilon)^{-1}=(p,q,\epsilon+p^\top q)$.
Fix any skew-symmetric $M\in\mathbb{F}_2^{k\times k}$.
For each $w\in\mathbb{F}_2^k$, we define $\psi_{M,w}\colon\mathbb{F}_2^k\to H_{2k+1}(\mathbb{F}_2)$ by
\[
\psi_{M,w}(p)
:=(p,Mp,Q_M(p)+w^\top p).
\]
One may apply the polarization identity
\begin{equation}
\label{eq.polarization}
Q_M(x+y)
=Q_M(x)+x^\top\tilde{M}y+y^\top\tilde{M}x+Q_M(y)
=Q_M(x)+x^\top My+Q_M(y).
\end{equation}
to verify that $\psi_{M,w}$ is a group homomorphism.

Next, define the \textit{Schr\"{o}dinger representation} $\rho\colon H_{2k+1}(\mathbb{F}_2)\to\operatorname{GL}(\mathbb{R}[\mathbb{F}_2^k])$ by
\[
\rho(p,q,\epsilon)
:=(-1)^\epsilon X_pZ_q,
\]
where $X_p$ and $Z_q$ denote translation and modulation operators, respectively:
\begin{equation}
\label{eq.def X and Z}
(X_pf)(x):=f(x+p),
\qquad
(Z_qf)(x):=(-1)^{q^\top x}f(x).
\end{equation}
Indeed, $\rho$ is a representation of $H_{2k+1}(\mathbb{F}_2)$ as a consequence of the easily verified relation
\begin{equation}
\label{eq.price of commuting}
Z_qX_p
=(-1)^{q^\top p}X_pZ_q.
\end{equation}
It follows that $\rho\circ\psi_{M,w}$ is a representation of the additive group $\mathbb{F}_2^k$.
Explicitly, we have
\begin{equation}
\label{eq.def Sp}
(\rho\circ\psi_{M,w})(p)
=(-1)^{Q_M(p)+w^\top p}X_pZ_{Mp}.
%=:S_p^{(M,w)}.
\end{equation}
Next, we put
\begin{equation}
\label{eq.def PMw}
P_{M,w}
:=2^{-k}\sum_{p\in\mathbb{F}_2^k}(\rho\circ\psi_{M,w})(p). 
%S_p^{(M,w)}.
\end{equation}
Then $P_{M,w}^2=P_{M,w}$ and $P_{M,w}^*=P_{M,w}$, and so $P_{M,w}$ is an orthogonal projection.
Decomposing into irreducible representations reveals that $\operatorname{im}P_{M,w}$ equals the intersection of the eigenspaces of $\{(\rho\circ\psi_{M,w})(p)\}_{p\in\mathbb{F}_2^k}$ with eigenvalue $1$.

We claim that $P_{M,w}=u_{M,w}u_{M,w}^*$.
To see this, one may first apply the definitions \eqref{eq.def Sp}, \eqref{eq.def X and Z}, and \eqref{eq.definition of vectors} along with \eqref{eq.polarization} and its consequence
\begin{equation}
\label{eq.isotropy}
x^\top Mx
=Q_M(x+x)+Q_M(x)+Q_M(x)
=0
\end{equation}
to verify the eigenvector equation
\[
[(\rho\circ\psi_{M,w})(p)]u_{M,w}=u_{M,w}
\]
%\[
%S_p^{(M,w)}u_{M,w}=u_{M,w}
%\]
for each $p\in\mathbb{F}_2^k$.
Next, one may apply the easily verified fact that
\begin{equation}
\label{eq.traces}
\operatorname{tr}(X_pZ_q)
=\left\{\begin{array}{cl}
2^k&\text{if }p=q=0\\
0&\text{otherwise}
\end{array}\right.
\end{equation}
to the definitions \eqref{eq.def PMw} and \eqref{eq.def Sp} to compute $\operatorname{tr}P_{M,w}=1$.
This proves our intermediate claim.

We are now ready to compute $|\langle u_{M,w},u_{M',w'}\rangle|^2$ in various cases.
First, we cycle the trace:
\[
|\langle u_{M,w},u_{M',w'}\rangle|^2
=\operatorname{tr}(u_{M,w}^*u_{M',w'}u_{M',w'}^*u_{M,w})
=\operatorname{tr}(u_{M,w}u_{M,w}^*u_{M',w'}u_{M',w'}^*)
=\operatorname{tr}(P_{M,w}P_{M',w'}).
\]
Next, we apply the definitions \eqref{eq.def PMw} and \eqref{eq.def Sp} along with identities \eqref{eq.price of commuting}, \eqref{eq.traces}, and \eqref{eq.isotropy} to obtain
\[
|\langle u_{M,w},u_{M',w'}\rangle|^2
=\operatorname{tr}(P_{M,w}P_{M',w'})
=2^{-k}\sum_{p\in\operatorname{ker}(M+M')}(-1)^{Q_M(p)+Q_{M'}(p)+(w+w')^\top p}.
\]
From here, we proceed in cases:
If $M=M'$, then $\operatorname{ker}(M+M')=\mathbb{F}_2^k$ and $Q_M(p)=Q_{M'}(p)$, and so $|\langle u_{M,w},u_{M,w'}\rangle|^2=\delta_{w,w'}$.
This gives (a).
For (b), if $M+M'$ has full rank, then $\operatorname{ker}(M+M')=\{0\}$, and so $|\langle u_{M,w},u_{M,w'}\rangle|^2=2^{-k}$, as claimed.
\end{proof}

\begin{proof}[Proof of Proposition~\ref{prop.kerdock construction}]
In what follows, we demonstrate three things:
\begin{itemize}
\item[(i)]
For every $s,x\in\mathbb{F}_{2^{k-1}}$ and $\alpha\in\mathbb{F}_2$, it holds that $(x,\alpha)\cdot L_s(x,\alpha)=0$.
\item[(ii)]
For every $s,x,y\in\mathbb{F}_{2^{k-1}}$ and $\alpha,\beta\in\mathbb{F}_2$, it holds that $(x,\alpha)\cdot L_s(y,\beta)=L_s(x,\alpha)\cdot (y,\beta)$.
\item[(iii)]
For every $r,s\in\mathbb{F}_{2^{k-1}}$ with $r\neq s$, $L_r(x,\alpha)=L_s(x,\alpha)$ implies $(x,\alpha)=(0,0)$.
\end{itemize}
Thanks to the non-degeneracy of the bilinear form, (i)--(iii) together imply the result.

First, (i) is easily verified by applying three properties of the trace: $\operatorname{tr}$ is $\mathbb{F}_2$-linear, $\operatorname{tr}(z^2)=\operatorname{tr}(z)$, and $\operatorname{tr}(z)^2=\operatorname{tr}(z)$ since $\operatorname{tr}(z)\in\mathbb{F}_2$.
Also, (ii) quickly follows from the linearity of the trace.
For (iii), take $r,s\in\mathbb{F}_{2^{k-1}}$ with $r\neq s$ and suppose $L_r(x,\alpha)=L_s(x,\alpha)$.
The second argument of this identity is $\operatorname{tr}(rx)=\operatorname{tr}(sx)$.
Rearrange the first argument to get
\[
0
=(r^2+s^2)x+r\operatorname{tr}(rx)+s\operatorname{tr}(sx)+\alpha(r+s)
=(r+s)\Big((r+s)x+\operatorname{tr}(rx)+\alpha\Big),
\]
where the last step applies the fact that $\operatorname{tr}(rx)=\operatorname{tr}(sx)$.
Since $r+s\neq0$ by assumption, it follows that
\begin{equation}
\label{eq.linear condition on x and alpha}
(r+s)x+\operatorname{tr}(rx)
=\alpha
\in\mathbb{F}_2.
\end{equation}
Since $\alpha^2=\alpha$, the left-hand side of \eqref{eq.linear condition on x and alpha} satisfies the same quadratic, which in turn implies $x\in\{0,(r+s)^{-1}\}$.
Since $x$ also satisfies $\operatorname{tr}((r+s)x)=0$, it follows that $x=0$.
Plugging into \eqref{eq.linear condition on x and alpha} then gives $\alpha=0$, as desired.
\end{proof}

\section{Numerical results}
\label{sec.numerics}

\begin{figure}[t]
\begin{tabular}{ccc}
$\|\hat{\mu}-Ax\|_\infty$
&
$\|\hat{h}-Ax\|_2$
&
ratio of runtimes
\\
\includegraphics[width=0.3\textwidth]{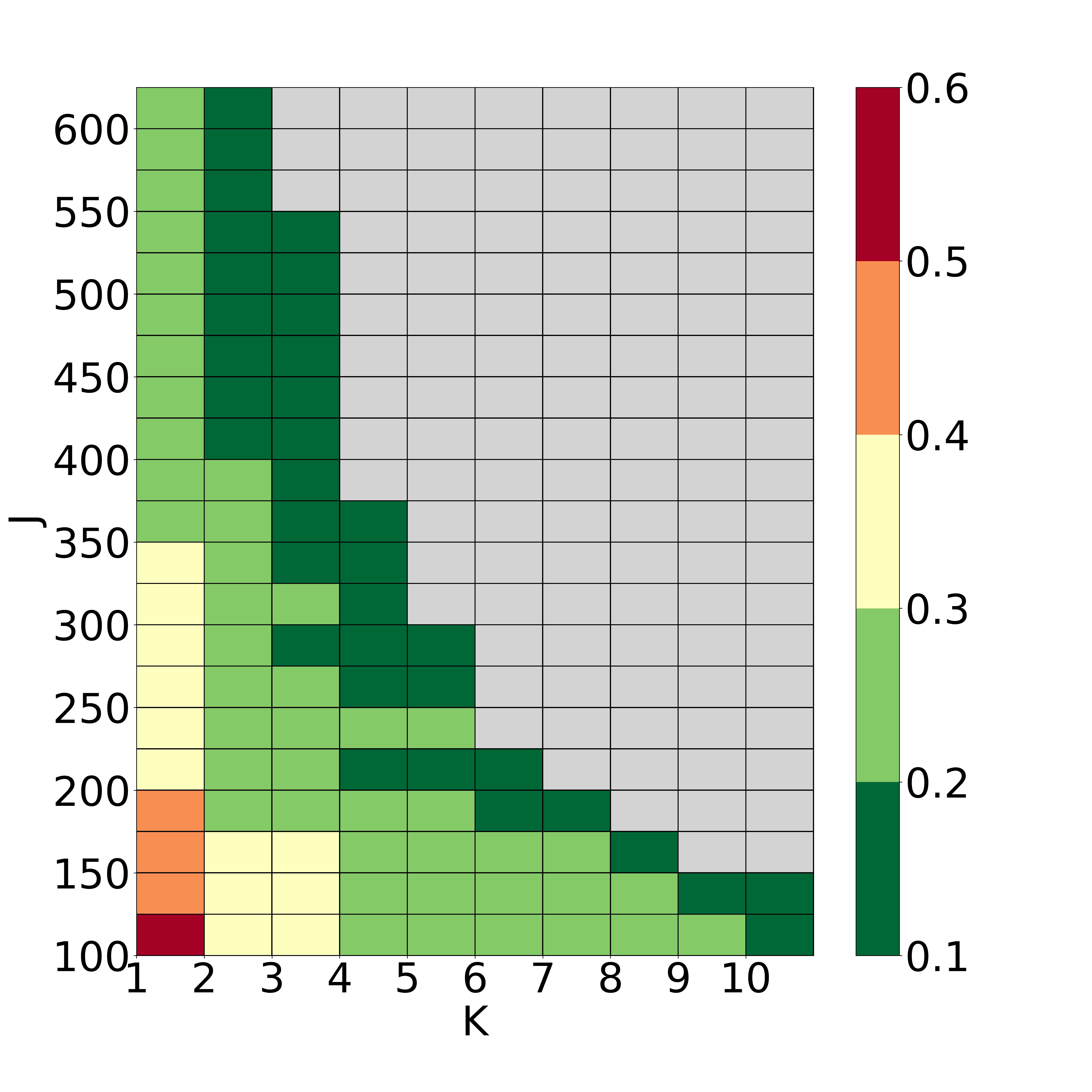}
&
\includegraphics[width=0.3\textwidth]{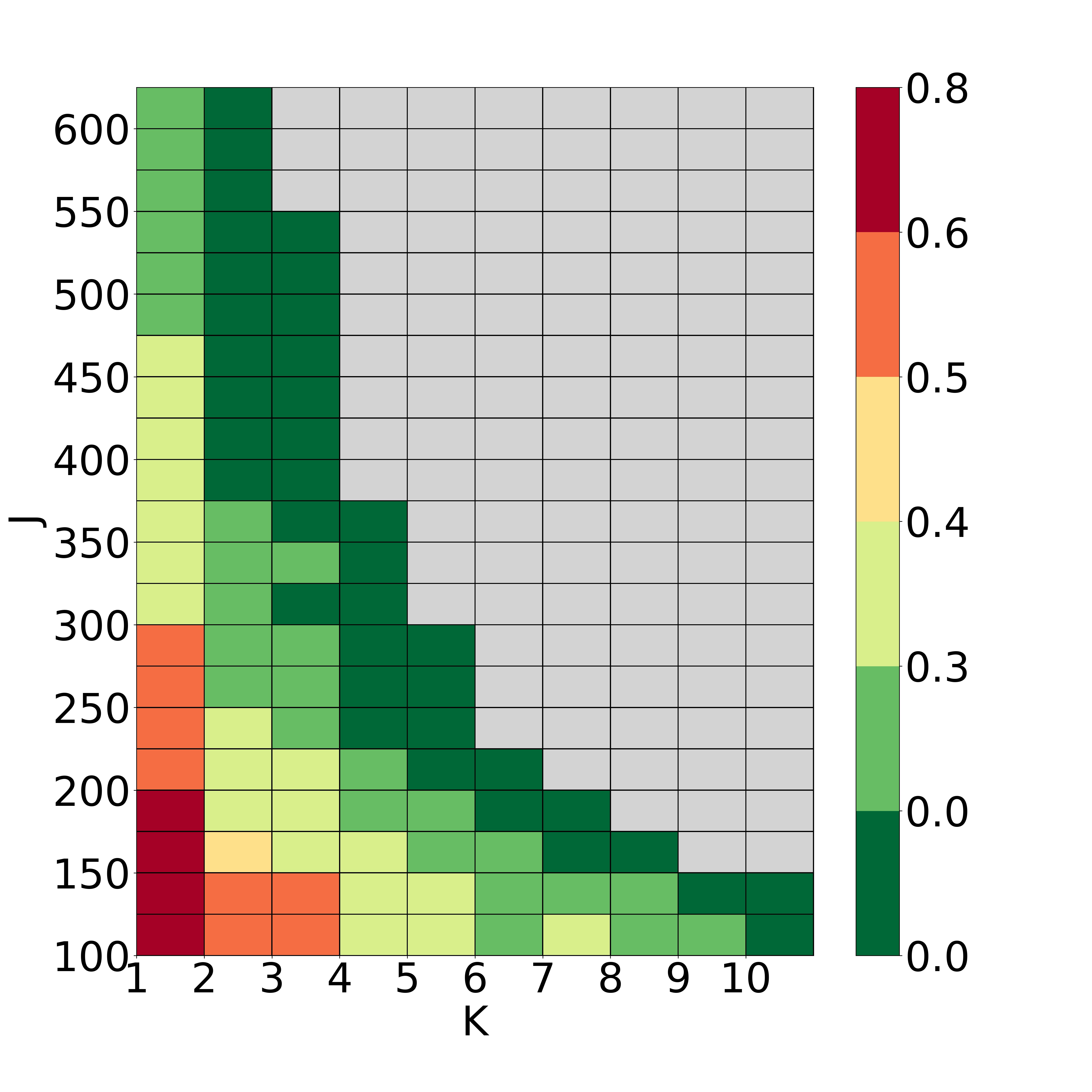}
&
\includegraphics[width=0.3\textwidth]{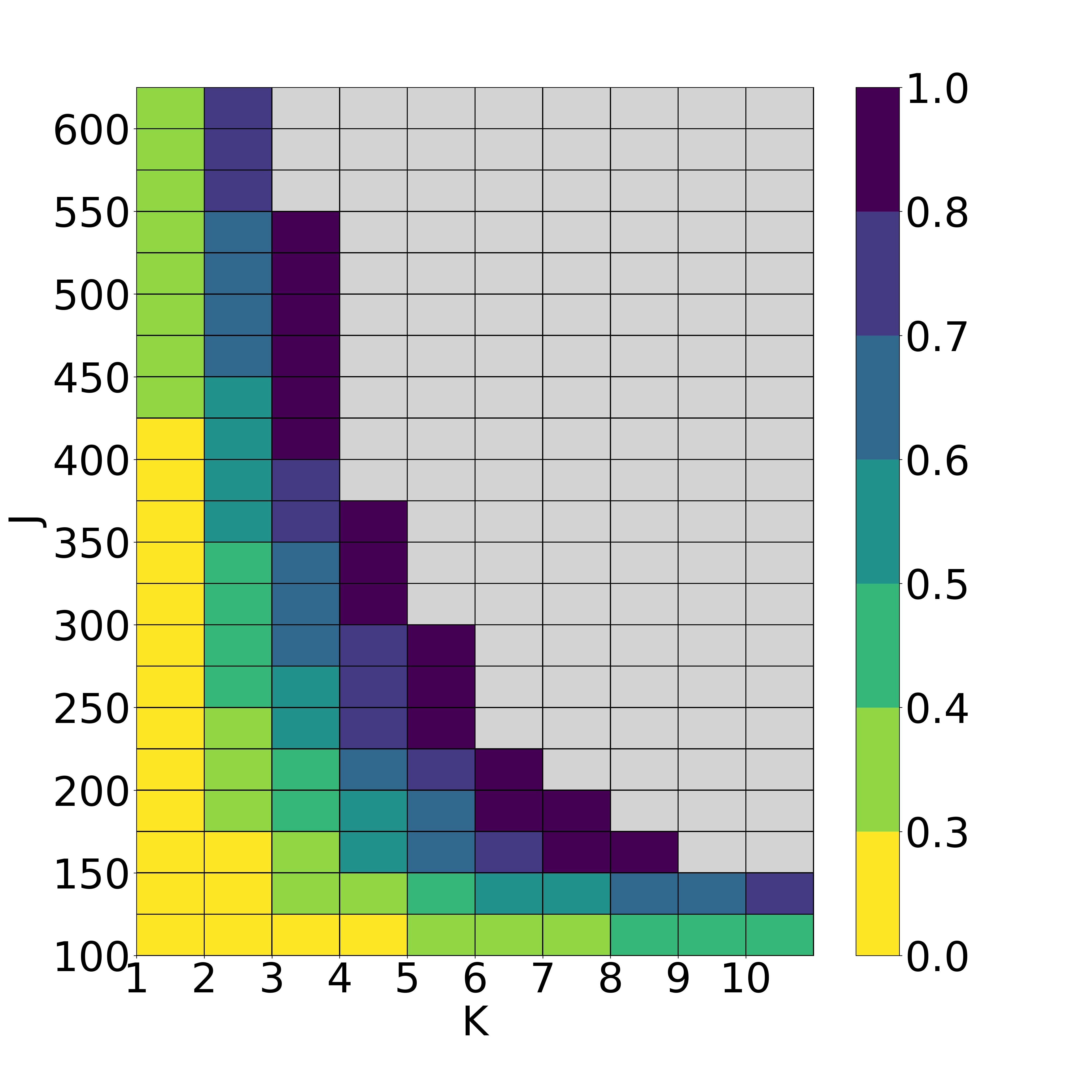}
\end{tabular}
\caption{
Performance of fast sparsifying transform for a random orthogonal matrix $A\in\mathbb{R}^{n\times n}$ with $n=4096$ and vectors $x\in\Sigma(A,s,\delta)$ with $s=20$ and $\delta=0$.
\textbf{(left)}
We compute the median-of-means estimator $\hat\mu\in\mathbb{R}^n$ and plot the worst-case behavior over $1000$ trials.
\textbf{(middle)}
We compute $\hat{h}\in\mathbb{R}^n$ by multiplying $x$ by the rows of $A$ corresponding to the $10s$ largest-magnitude entries of $\hat\mu$, and we plot the worst-case behavior over $1000$ trials.
\textbf{(right)}
We plot the quotient of our algorithm's runtime with the runtime of naive matrix--vector multiplication.
(Here, we ignore the runtime of randomly sampling $N$ vectors from the precomputed sketch $\{As_\ell\}_{\ell\in[L]}$ since optimizing data structures is beyond the scope of this paper.)
Throughout, gray denotes choices of parameters for which our algorithm is no faster than the naive algorithm.
}
\label{figure}
\end{figure}

In this section, we report the real-world performance of our fast sparsifying transform.
In our experiments, we take $A\in\mathbb{R}^{n\times n}$ to be a random orthogonal matrix, and then we select a unit vector $x$ such that $Ax$ has exactly $s$ nonzero entries of the same size in random positions.
(This is straightforward to implement since $A^{-1}=A^\top$.)
Due to limitations in computing power and storage capacity, we restrict our experiments to dimension $n=4096$, and we select sparsity level $s=20$.

What follows are some details about our implementation.
Since $n$ is a power of $2$, we have $d=n$, and so our projective $2$-design has size $L=n(n/2+1)=8392704$.
For such a large value of $L$, it turns out that the runtime of selecting $N$ random members of the precomputed sketch $\{As_\ell\}_{\ell\in[L]}$ is sensitive to the design of the underlying data structure.
In order to provide a useful runtime comparison, we therefore assume that this random selection is performed by an oracle before we start the runtime clock in our algorithm.
To compute medians, we apply the quickselect algorithm~\cite{MahmoudHM:95}.
Finally, to boost performance, we take $S\subseteq[n]$ to index the $10s$ largest-magnitude entries of $\hat\mu$ instead of the top $s$ entries.

The results of our experiments are summarized in Figure~\ref{figure}.
Figure~\ref{figure}(left) illustrates that the median-of-means estimator $\hat\mu$ performs better in practice than predicted by Theorem~\ref{thm.median of means}.
In particular, we can take $J$ and $K$ to be smaller than suggested by the bounds in our guarantee, which is good for runtime considerations.
(In fact, taking $\eta=1$ in Theorem~\ref{thm.median of means} delivers a lower bound on $K$ that is greater than $16$, meaning our theoretical guarantees are off the scale in this plot.) 
Figure~\ref{figure}(middle) illustrates the performance of our entire algorithm for different choices of $J$ and $K$.
Notably, we perfectly computed $Ax$ in all of our $1000$ random trials when $K=2$ and $J=375$.
Figure~\ref{figure}(right) illustrates the runtime of our algorithm relative to naive matrix--vector multiplication.
For this plot, we divide the runtime of our algorithm by the runtime of the naive algorithm.
Throughout, gray denotes choices of $(K,L)$ for which our algorithm provides no speedup over the naive algorithm.
In particular, when $K=2$ and $J=375$, our method is about twice as fast as the naive approach.
Interestingly, the median-of-means estimator reduces to the empirical mean when $K=2$; this might suggest an opportunity to improve our theory, and perhaps even speed up our algorithm.

\section*{Acknowledgments}

We are grateful to Peter Jung for interesting discussions.
TF and FK were partially supported by the German Research Council (DFG) via contract KR 4512/2-2, DG was partially supported by the German Research Council (DFG) via contract GR4334/1-1, and DGM was partially supported by AFOSR FA9550-18-1-0107 and NSF DMS 1829955.

\end{document}